\documentclass[draft,12pt,reqno,a4paper]{amsart}

\usepackage[margin=3cm,footskip=1cm]{geometry}

\usepackage{amsmath}
\usepackage{amsfonts}
\usepackage{amsthm}
\usepackage{mathtools}

\usepackage{enumerate}
\usepackage[latin1]{inputenc}
\usepackage[shortlabels]{enumitem}

\usepackage{graphicx}
\usepackage{verbatim}

\DeclareMathOperator*{\wslim}{w^\star--lim}

     \newcommand{\supp}{\operatorname{supp}}

     \newcommand{\Ran}{{\operatorname{Ran}}}

\newcommand{\Opw}{{\operatorname{Op\!^w}}}
     \newcommand{\N}{{\mathbb{N}}}

     \newcommand{\R}{{\mathbb{R}}}
     \newcommand{\Z}{{\mathbb{Z}}}
     \newcommand{\C}{{\mathbb{C}}}

\newcommand{\e}{{\rm e}}

\renewcommand{\i}{{\rm i}}
\renewcommand{\d}{{\rm d}}

\newcommand{\righ}{{\rm right}}
\newcommand{\lef}{{\rm left}}
\newcommand{\unif}{{\rm unif}}

\renewcommand{\Re}{{\rm Re}\,}
\renewcommand{\Im}{{\rm Im}\,}

\newcommand\inp[2][]{#1 \langle #2#1\rangle}
\newcommand\parb[2][]{#1 \big ( #2#1\big )}

\newcommand{\pp}{{\rm pp}}

\renewcommand{\exp}{{\rm exp}}
\newcommand{\mand}{\text{ and }}
\newcommand{\mfor}{\text{ for }}
\newcommand{\mforall}{\text{ for all }}

\newcommand{\vB}{{\mathcal B}}

\newcommand{\vH}{{\mathcal H}}

     \theoremstyle{plain}
     \newtheorem{thm}{Theorem}[section]
     \newtheorem{proposition}[thm]{Proposition}
     \newtheorem{lemma}[thm]{Lemma}
      \newtheorem{corollary}[thm]{Corollary}
     \theoremstyle{definition}

     \newtheorem{cond}[thm]{Condition}

     \newtheorem{remark}[thm]{Remark}

\newtheorem*{remarks*}{Remarks}
\newtheorem*{remark*}{Remark}

     \numberwithin{equation}{section}

\title[Sommerfeld radiation condition at threshold]{Sommerfeld radiation condition at threshold}

\author{E. Skibsted}
\address[E. Skibsted]{Institut for  Matematiske
Fag \\
Aarhus Universitet\\ Ny Munkegade  8000 Aarhus C,
Denmark}
\email{skibsted@imf.au.dk}

\date{}
\begin{document}
\maketitle
\thanks{\textit{Dedicated to Hiroshi Isozaki on the occasion of his
    Sixtieth Birthday.}}

\begin{abstract}
We prove Besov space bounds of the resolvent at low  energies in any
dimension for a
class of potentials  that are negative  and obey a virial
condition with these conditions imposed  at infinity only. We do
not require  spherical symmetry. The class of potentials  includes  in
dimension $\geq3$  the attractive
Coulomb potential. There are  two boundary values of the resolvent at
zero energy which we  characterize by radiation conditions. These
radiation conditions are zero energy versions of  the well-known Sommerfeld radiation condition. 
\end{abstract}

\tableofcontents

\newpage

\section{Introduction}\label{sec:introduction}
In this  paper we study low-energy   spectral  theory for the
Schr\"odinger operator
$H = -\Delta + V$  on $\mathcal H=L^2({\mathbb R}^d),\, d\geq1$, where
the potential
$V$ obeys the following condition. We  
use the notation  $\langle x \rangle = \sqrt{ x^2+1}$,
$\N_0=\N\cup\{0\}$, and for $\mu \in (0,2)$ the  notation  $s_0=1/2+\mu/4$. 
\begin{cond}
\label{cond:lap}
Let $V=V_1+V_2$ be a
real-valued function defined on  ${\mathbb R}^d$; $
d\geq1$.
There exists  $\mu \in (0,2)$  such
that the following  conditions
{\rm
\ref{it:assumption1}--\ref{it:assumption_last}} hold.

\begin{enumerate}[\normalfont (1)]
   \item \label{it:assumption1} There exists $\epsilon_1 > 0$ such
     that $V_1(x) \leq -\epsilon_1
\langle x \rangle^{-\mu}.$
   \item \label{it:assumption2}     $V_1\in C^\infty(\R^d)$. For all $\alpha\in \N_0^{d}$
there exists
$C_{\alpha} >0$ such that
$$
\langle x \rangle^{\mu+|\alpha|} |\partial^{\alpha} V_1(x)|
\leq C_{\alpha}.$$
\item \label{it:assumption3} 
     There exists $\tilde \epsilon_1 > 0$ such that $-|x|^{-2}
\left(x\cdot \nabla (|x|^2 V_1)\right)
\geq -\tilde \epsilon_1 V_1$.
   
   \item \label{it:assumption4} There exists $\delta,C,R > 0$
such that 
$$
                                |V_2(x)| \leq C |x|^{-2s_0-\delta},
$$ 
for $|x| > R$. 
   \item \label{it:assumption_last} 
     $V_2\in L^2_{\rm loc }(\R^d)$ for $d=1,2,3$,  $V_2\in L^p_{\rm
       loc }(\R^d)$ for some $p>2$ if  $d=4$ while $V_2\in
     L^{d/2}_{\rm loc }(\R^d)$ for $d\geq 5$.
\end{enumerate}
\end{cond}
Due to \ref{it:assumption4} and \ref{it:assumption_last} the operator $V_2 (-\Delta +i)^{-1}$ is a compact operator on
     $L^2({\mathbb R}^d)$, see for example \cite[Theorems X.20 and X.21]{RS} for
     the case $d\geq 4$. Whence  $H$ is self-adjoint. The
Schr\"odinger operator with  an  attractive
Coulomb potential in
dimension $d\geq3$  is a particular example.

While low-energy   spectral  asymptotics for Schr\"odinger operators is a well studied subject for
classes of potentials of fast decay the literature is more sparse for
classes of potentials with decay $O(r^{-2})$ or slower. We refer to
\cite{Ya,Na,FS,SW} and references therein. We remark that Condition \ref{cond:lap} is closely
related to the conditions used in \cite{FS}, in fact in the present
paper  we aim
at proving more precise  resolvent bounds than done in \cite{FS} (Besov space bounds), and
characterize boundary values $R(0\pm \i0)$ of
the resolvent at zero energy. The latter  is achived in terms of certain microlocal
estimates traditionally referred to as Sommerfeld radiation
  conditions. For positive energies the
  limiting absorption principle (LAP) and Sommerfeld radiation conditions  are fundamental for stationary
  scattering theory and they are used at many places in the
  literature, see for example  \cite{Sa,
    AH, GY} and \cite[Section 30.2] {Ho2} (two-body problems), \cite{Is2,Va} (many-body problems) and 
  \cite{Mol} (abstract framework). Moreover  radiation conditions are
  an 
  integral part of for one of the oldest
  method of proving LAP, cf. for example \cite{Sa}. We consider them
  as interesting of their own right, in particular including the case
  of zero energy cf. \cite{DS3}.

Neither  \cite{FS} nor the present paper  deal with
scattering theory. On the other hand  there are indeed applications to scattering theory at low
energies   as demonstrated in the recent works \cite{DS3,DS2}. However this  is for a smaller class of potentials
(essentially radially symmetric potentials).  We plan in  a future
publication \cite{Sk}  to study scattering theory at low
energyies for a larger class than the
ones of \cite{DS3,DS2} however within the one defined by Condition
\ref{cond:lap}. For this study  Besov space  bounds  and uniqueness
induced by versions of the 
 Sommerfeld radiation condition  will be 
useful. We 
remark that these results are
also present in some form in  \cite{DS3} although, as indicated above, this is under
stronger conditions on the potentials. Besides  the Besov space
bounds are
not shown in the strongest form  as done here  and they are obtained
somewhat indirectly (in fact only the imaginary part
of the boundary value of the 
resolvent is estimated). We present (most of) our results in Subsection
\ref{sec:results}. They are all under Condition
\ref{cond:lap}.
 
\subsection{Results}\label{sec:results}

\subsubsection{Resolvent bounds}\label{sec:resolv-bounds-sectFS}

Let us recall  a main result from \cite{FS}.
Let $\theta \in (0, \pi)$, $\lambda_0>0$  and define
\begin{equation}
  \label{eq:def_Gamma_theta}
\Gamma_{\theta} = \{ z \in {\mathbb C}\setminus \{0\}
\,\big{ | } 
\arg z \in (0,\theta) , \,
|z| \leq \lambda_0\}.
\end{equation} In \cite{FS} $\lambda_0$ is exclusively taken equal to
one although this is only for convenience of presentation. In this
paper we fix any $\lambda_0>0$ at this point and suppress henceforth any
 dependence of this constant (as done in the notation
 \eqref{eq:def_Gamma_theta}). At this point we also fix $\theta \in
 (0, \pi)$, but keep (somewhat inconsistently) the dependence of
 $\theta$ of the set \eqref{eq:def_Gamma_theta}.

For $\mu \in (0,2)$, $K>0$ and $ \lambda \geq 0$ let 
\begin{equation}
  \label{eq:fsublambda}
  f=f_\lambda(x)=(\lambda+K\inp{x}^{-\mu})^{1/2};\, x\in \R^d.
\end{equation}
Here $\lambda$ will be taken  as $|z|$ for $z$ in the closure
of  $\Gamma_{\theta}$ and $K$ can for parts of our presentation and analysis be taken
arbitrary. More precisely the latter is true for Theorems \ref{thm:lap}
and  \ref{prop:resolvent-bounds} presented below. Consequently we take, for
convenience, $K=1$ in   these theorems as well as in Subsection \ref{subsec: Concrete Besov
spaces} (where Theorem \ref{prop:resolvent-bounds} is proved).  As for Section \ref{subsec: Sommerfeld radiation condition
at zero energy}  we choose a different value of $K$,
see \eqref{eq:26}.

For a Hilbert space ${\mathcal H}$ (which in our case  will be
$L^2({\mathbb R}^d)$) we denote by ${\mathcal B}(\mathcal H)$ the
space of bounded linear operators on ${\mathcal H}$  
(a similar notation will be used for Banach spaces).
A ${\mathcal B}(\mathcal H)$-valued function $T(\cdot)$ on
$\Gamma_{\theta}$ is said to be uniformly H\"{o}lder continuous in 
$\Gamma_{\theta}$ if there exist $C,\gamma>0$ such that 
$$
\|T(z_1)-T(z_2) \|\leq C|z_1-z_2|^{\gamma}\; {\rm{for\; all}}\; z_1,z_2\in
\Gamma_{\theta}.
$$

We consider the Schr{\"o}dinger operator $H = -\Delta + V$ on
$L^2({\mathbb R}^d)$ under Condition~\ref{cond:lap}. The resolvent is
denoted by $R(z) = (H -z)^{-1}$.  In the statement below of (a version
of) \cite[Theorem 1.1]{FS} some conditions on the potential $V$ are
slightly changed. Comments at this point are given after the
statement.

\begin{thm}[LAP]
\label{thm:lap}
Suppose Condition \ref{cond:lap}.
 For all  $s>s_0$
  the family of operators $T(z)=\langle x
\rangle^{-s}R(z) 
\langle x \rangle^{-s}$ is uniformly H\"{o}lder continuous in 
$\Gamma_{\theta}$. 
In particular the 
limits
\begin{align*}
T(0 +\i0)&=\langle x \rangle^{-s} R(0 + \i0) 
\langle x \rangle^{-s}
= \lim_{z \rightarrow 0, z \in \Gamma_{\theta}}
T(z),\\
T(0 -\i0)&=\langle x \rangle^{-s} R(0 -\i0) \
\langle x \rangle^{-s}
= \lim_{z \rightarrow 0, z \in \Gamma_{\theta}}
T(\bar z)
\end{align*}
exist in
${\mathcal B}(L^2({\mathbb R}^d))$.

For all $s> 1/2$  there exists $C>0$
such that for all $z \in \Gamma_{\theta}$
\begin{equation}
  \label{eq:limitbound}
 \|
\langle x
\rangle^{-s} f_{|z|}^{1/2 }R(z)
f_{|z|}^{1/2 }\langle x \rangle^{-s}\| \leq C.
\end{equation}
\end{thm}

 We have already noted that due to \ref{it:assumption4} and
 \ref{it:assumption_last} the operator $V_2 (-\Delta +i)^{-1}$ is a compact, which occurs as a separate condition in \cite[Theorem
     1.1]{FS} (the condition \ref{it:assumption_last} is a new condition compared to
 \cite[Theorem 1.1]{FS}). Another condition from
     \cite[Theorem 1.1]{FS} that we omitted above is a version of  unique continuation
     at infinity. This version, \cite[Assumption 2.1]{FS},  is automatically satisfied, given
     \ref{it:assumption_last}, due to results of \cite{FH} (for $d\leq
     3$) and
     \cite{JK} (for $d\geq
     3$). Applying it with $V\to V-\lambda$ for any $\lambda \geq0$ in conjunction
     with 
     \cite[Theorem 2.4]{FS} we
     have $\sigma_{\pp}(H)\cap [0,\infty)=\emptyset$ (for $d=1,2,3$
     absence of strictly positive
     eigenvalues follows alternatively from \cite[Corollary 1.4]{FH}). The absence of
     non-negative eigenvalues is of course  a
     consequence of \eqref{eq:limitbound}, however  this property is part of the
     proof of the latter  bound. 

     We note that imposing the conditions \ref{it:assumption1} and
     \ref{it:assumption3} only near infinity may seem, with the other
     conditions of Condition \ref{cond:lap}, to weaken the
     assumptions. However this is not the case cf. a discussion in
     \cite[Section 3]{FS}. On the other hand it suffices to have the
     bounds \ref{it:assumption2} for $|\alpha|\leq 2$. More precisely
     this is the case for Theorems~\ref{thm:lap}
     and~\ref{prop:resolvent-bounds}. For the microlocal estimates of
     Section \ref{subsec: Sommerfeld radiation condition at zero
       energy} though we need $V_1$ to be a ``symbol'' and all bounds
     of \ref{it:assumption2} are then needed.

Notice also that \eqref{eq:limitbound} is stronger than 
boundedness of the family $T(\cdot )$ (which is a consequence of the
uniform H\"{o}lder continuity). 

A main  result of this paper (recall that we impose Condition
\ref{cond:lap} throughout the paper) is the following improvement of \eqref{eq:limitbound} in
terms of Besov spaces as defined in the beginning of Subsection \ref{subsec: Abstract Besov
spaces}  with the operator $A$ there being given as multiplication by
$|x|$ on the complex Hilbert space $\vH=L^2({\mathbb R_x^d})$.
\begin{thm}[Besov space bound]
  \label{prop:resolvent-bounds}  There exists
  $C>0$ such that for all $z \in \Gamma_{\theta}$
\begin{equation}
  \label{eq:limitbound2j}
 \|
f_{|z|}^{1/2 }R(z)
f_{|z|}^{1/2 }\|_{\vB (B(|x|),B(|x|)^*)} \leq C.
\end{equation}
\end{thm}

\subsubsection{Sommerfeld radiation condition}\label{sec:Sommerfeld
  radiation condition} We shall  give an outline  of the
results of Section \ref{subsec: Sommerfeld radiation
  condition at zero energy}. These results are on microlocal estimates
of  solutions to the equation $Hu=v$. In particular we estimate and
characterize the particular solution provided by Theorems
\ref{thm:lap} and \ref{prop:resolvent-bounds}. This particular
solution is constructed as follows in terms of Besov
spaces. First note that the  relevant Besov space at zero energy  is
$B^\mu:=\inp{x}^{-\mu/4}B(|x|)$, cf.  Theorem
\ref{prop:resolvent-bounds}. We have the following characterization of the
corresponding dual
space (recall $s_0:=1/2+\mu/4$)
\begin{equation*}
  u\in (B^
\mu)^* \Leftrightarrow  u\in L^2_{\rm loc}(\R^d)\mand \sup_{R>1} R^{-s_0}\|F(|x|<R)u\|<\infty.
\end{equation*} A slightly smaller space is given by 
\begin{equation*}
  u\in (B^
\mu)^*_0 \Leftrightarrow u\in  L^2_{\rm loc}(\R^d) \mand \lim_{R\to \infty} R^{-s_0}\|F(|x|<R)u\|=0. 
\end{equation*}
Now suppose $v\in B^ \mu$.  Then due to Theorems \ref{thm:lap} and
\ref{prop:resolvent-bounds} there exists the weak-star limit
\begin{equation}\label{eq:30o}
  u=R(0 + \i0) v
= \wslim_{z \rightarrow 0, z \in \Gamma_{\theta}}
R(z) v\in (B^
\mu)^*. 
\end{equation}  Note that indeed this $u$ is a (distributional)
solution to the equation $Hu=v$. 

 Let us state a  microlocal property of  this solution. We shall use
 \eqref{eq:fsublambda} with  
\begin{equation}
  \label{eq:26}
  K=\epsilon_1 \tilde \epsilon_1 /(2-\mu),
\end{equation} where the $\epsilon$'s come from Condition
\ref{cond:lap}. 
 In terms of  $f_0$ we then introduce symbols 
\begin{align*}
a_0= \frac{\xi^2}{f_0(x)^2} ,\;\;
b_0=  \frac{\xi}{f_0(x)} \cdot \frac{x}{\langle x \rangle},
\end{align*} and we prove that 
\begin{equation}\label{eq:31h}
  \Opw(\chi_-(a_0)\tilde \chi_-(b_0))u\in (B^
\mu)^*_0\mforall \chi_-\in C^\infty_c(\R)\mand \tilde \chi_-\in C^\infty_c((-\infty, 1)).
\end{equation} Here we use Weyl quantization (although this is not the 
only choice). For \eqref{eq:30o}, \eqref{eq:31h}, another  version of 
\eqref{eq:31h} as well as   another microlocal property (high energy
estimates) we refer the
reader to  Proposition \ref{prop:radiation_somm-radi-conds}. 

The support property of $\tilde \chi_-$ in \eqref{eq:31h}  mirrors that
the particular solution  studied   is ``outgoing'', and we refer to
\eqref{eq:31h}  as  a 
\emph{Sommerfeld radiation condition}. This condition   (in fact a
weaker version) suffices   for a characterization as expresssed in the
following result. We refer the reader to Theorem
\ref{thm:somm-radi-conds} for a  slightly stronger result as well as
another 
version.
\begin{thm}[Uniqueness of outgoing
  solution]\label{thm:somm-radi-condss}  Suppose   $v\in  B^
\mu$. Suppose $u$ is a  distributional 
solution to the equation $Hu=v$ belonging to the space $
\inp{x}^{-s}L^2(\R^d)$ for some $s\in \R$, and suppose that there
there exists $\sigma\in(0,1]$ such that 
\begin{equation}\label{eq:31hf}
  \Opw(\chi_-(a_0)\tilde \chi_-(b_0))u\in (B^
\mu)^*_0\mforall \chi_-\in C^\infty_c(\R)\mand \tilde \chi_-\in C^\infty_c((-\infty, \sigma)).
\end{equation} Then $u=R(0 + \i0) v$. In particular \eqref{eq:31h} holds.
\end{thm}

The proof of Theorem \ref{thm:somm-radi-conds} (yielding in particular
Theorem
\ref{thm:somm-radi-condss}) relies partly on a ``propagation
of singularities'' result. This result is  stated as Proposition
\ref{prop:propa_sing}. We note that the ``incoming'' solution  $u=R(0
- \i0) v$ can be   characterized  similarly. Our results generalize
\cite[Proposition 4.10] {DS3} at zero energy. For similar results for
positive energies and for   larger classes of potentials see
\cite[Theorem 30.2.10] {Ho2} and \cite{GY}.

 \section{Improved resolvent bounds}\label{sec:Resolvent bounds}
In Subsection \ref{subsec: Concrete Besov
spaces}  we prove Theorem \ref{prop:resolvent-bounds}.  The proof will be based on various results for abstract Besov
spaces to be given in Subsection \ref{subsec: Abstract Besov
spaces}.

\subsection{ Abstract Besov
spaces}\label{subsec: Abstract Besov
spaces}

Let $A$ be a self-adjoint operator on a Hilbert space~$\vH$. Let
$R_0=0$ and $R_j=2^{j-1}$ for $j\in \N$.  We define correspondingly
characteristic functions $F_j=F(R_{j-1}\leq |\cdot|<R_j)$ and the
space
\begin{equation}
  \label{eq:1}
  B=B(A)=\big \{u\in \vH \big | \,\sum_{j\in\N}R_j^{1/2}\|F_j(A)u\|=:\|u\|_B<\infty\big \}.
\end{equation}  We can identify (using the embeddings $\inp{A}^{-1}\vH\subseteq
B\subseteq \vH \subseteq B^*$, $\inp{A}:=\sqrt{A^2+1}$\,)
 the dual  space $B^*$ as 
\begin{equation}
  \label{eq:2}
  B^*=B(A)^*=\big \{u\in \inp{A}\vH \big | \,\sup_{j\geq 1} R_j^{-1/2}\|F_j(A)u\|=:\|u\|_{B^*}<\infty\big \}.
\end{equation} Alternatively, the elements  $u$  of $B^*$ are  those {\it
  sequences} $u=(u_j)\subseteq  \vH $ with $u_j\in \Ran
(F_j(A))$ and $\sup_{j\in\N} R_j^{-1/2}\|u_j\|<\infty$. This abstract
space was also considered in \cite{JP} (note however that $B(A)^*$  
 is identified incorrectly  in  \cite{JP} as   the completion  of $\vH$
in the norm $\|\cdot\|_{B^*}$). For other previous related works we refer to
\cite{AH,GY,Wa,Ro} and \cite[Subsections 14.1 and 30.2]{Ho2}. We note the
bounds, cf. \cite[Subsections 14.1]{Ho2},
\begin{equation}
  \label{eq:2p}
  \|u\|_{B^*}\leq \sup_{R>1} R^{-1/2}\|F(|A|<R)u\|\leq 2\|u\|_{B^*}. 
\end{equation}

Introducing \emph{abstract   weighted spaces} $L^2_s=L^2_s(A)=\inp{A}^{-s}\vH $ and
$B^*_0=B^*_0(A)$, the completion  of $\vH$ in the space $B^*$,  we have the embeddings
\begin{equation}\label{eq:3}
  L^2_s\subseteq B\subseteq L^2_{1/2}\subseteq \vH \subseteq
  L^2_{-1/2}\subseteq B^*_0\subseteq B^*\subseteq L^2_{-s},\mforall s>1/2.
\end{equation} All embeddings are continuous and corresponding
bounding constants can be chosen as absolute constants,
i.e. independently of $A$  and $\vH$. In particular
\begin{equation}
  \label{eq:22}
  \|u\|_{\vH}\leq \|u\|_B \mforall u\in B.
\end{equation}
Note also that
\begin{equation}
  \label{eq:41}
  u\in B^*_0 \text{ if and only if  }u\in B
^*\mand \lim_{R\to \infty} R^{-1/2}\|F(|A|<R)u\|=0. 
\end{equation}

We refer to the spaces $B,B^*$ and $B^*_0$ as \emph{abstract Besov
  spaces}.  Recall the following interpolation type result, here
stated abstractly. The proof is the same as that of the concrete
versions \cite[Theorem~2.5]{AH}, \cite[Theorem~14.1.4]{Ho2},
\cite[Proposition~2.3]{JP} and \cite[Subsection~4.3]{Ro}.
\begin{lemma}
  \label{lemma:resolvent-bounds} Let $A_1$ and  $A_2$ be self-adjoint
  operators on Hilbert spaces $\vH_1$ and $\vH_2$,
  respectively, and let  $s>1/2$. Suppose $T\in \vB(\vH_1,\vH_2)\cap
  \vB(L^2_s(A_1),L^2_s(A_2))$. Then $T\in \vB(B(A_1),B(A_2))$,
  and there is a constant $C=C(s)>0$ (independent  of $T$) such that
  \begin{equation}
    \label{eq:4}
    \|T\|_{\vB(B(A_1),B(A_2))}\leq
    C\parb{\|T\|_{\vB(\vH_1,\vH_2)}+\|T\|_{\vB(L^2_s(A_1),L^2_s(A_2))}}.
  \end{equation}
\end{lemma}
\begin{corollary}
  \label{cor:abstr-besov-spac}
 Let $A_1$ and  $A_2$ be self-adjoint
  operators on a Hilbert spaces $\vH$. Suppose that
  $\inp{A_2}^{s}\inp{A_1}^{-s}\in\vB(\vH)$ for some  $s>1/2$. Then
  $B(A_1)\subseteq B(A_2)$ and \begin{equation}
    \label{eq:4jjjj}
    \|u\|_{B(A_2)}\leq  C \|u\|_{B(A_1)}\mforall u\in B(A_1).
  \end{equation}
\end{corollary} 

 The norm on $B(A)$ in \eqref{eq:1} is not the only possible choice: Define for
 any  $p>1$

\begin{equation}
  \label{eq:1b}
  B_p=B(A)_p=\big \{u\in \vH \big | \,\sum_{j\in\N}\tilde
   R_j^{1/2}\|\tilde F_j(A)u\|=:\|u\|_{B_p}<\infty\big
  \},
\end{equation} where $\tilde F_j=F(\tilde R_{j-1}\leq |\cdot|<\tilde
R_j)$,  $\tilde R_0=0$ and $\tilde R_j=p^{j-1}$ for $j\geq1$. For $p=2$ this agrees with \eqref{eq:1}.

\begin{lemma}
  \label{lemma:resolvent-boundsoofirst} Let $A$ be   a self-adjoint  operator  on a Hilbert
  space $\vH$. For all  $p>1$ the space $B(A)_p=B(A)$ and  there exists
  $C=C(p)>0$ (i.e. independent  of $A$) such that  for all $u\in B(A)$
  \begin{equation}
    \label{eq:4ooBfirst}
    C^{-1}\|u\|_{B(A)_p}\leq \|u\|_{B(A)}\leq  C\|u\|_{B(A)_p}.
  \end{equation}
\end{lemma}
\begin{proof}
  The first term in the expression $\sum_{j\in\N}\tilde
   R_j^{1/2}\|\tilde F_j(A)u\|$ is bounded by
   $\|u\|_{B(A)}$, cf. \eqref{eq:22}. So let us look at a term with
   $j\geq2$. We estimate
   \begin{equation*}
  \tilde R_j^{1/2}\|\tilde F_j(A)u\|\leq \sqrt p\|\tilde
  F_j(A)|A|^{1/2}u\|\leq \sqrt p\sum^{\infty}_{k=2}  \|\tilde
  F_j(A)F_k(A)|A|^{1/2}u\|.
   \end{equation*} Now a small consideration shows that for all $k\geq
   2$ there are at most $2+[\ln 2/\ln p]$ number of $j$'s with $j\geq 2$  for which $\tilde
  F_jF_k\neq 0$. Whence 
  \begin{align*}
  \MoveEqLeft
  \sum_{j\geq 2}\tilde
   R_j^{1/2}\|\tilde F_j(A)u\|
   \\
   &\leq \sqrt p \parb{2+\ln 2/\ln p} \sum^{\infty}_{k=2}  \|
  F_k(A)|A|^{1/2}u\|\leq \sqrt p \parb{2+\ln 2/\ln p}\|u\|_{B(A)},
  \end{align*} yielding the first inequality in
  \eqref{eq:4ooBfirst}  for any $C\geq 1+\sqrt p \parb{2+\ln 2/\ln
    p}$. 

By the same method one shows the second  inequality in
  \eqref{eq:4ooBfirst} for any $C\geq 1+\sqrt 2 \parb{2+\ln p/\ln
    2}$. 
\end{proof}
\begin{lemma}
  \label{lemma:resolvent-boundsoo} Let   $A$ be  operator  on a Hilbert
  space $\vH$, such that $A\geq I$, and $s>-1$ be  given. Then $A^{-s/2}:B(A)\to B(A^{1+s})$
  is a homeomorhic isomorphism,  and there is a constant $C=C(s)>0$ (i.e. independent  of $A$) such that
  \begin{equation}
    \label{eq:4oo}
    \|A^{-s/2}\|_{\vB(B(A),B(A^{1+s}))}\leq  C\mand \|A^{s/2}\|_{\vB(B(A^{1+s}),B(A))}\leq  C.
  \end{equation}
\end{lemma}
\begin{proof}
  We let $p=2^{1/(1+s)}$, $C=2^{\max(s/2,\,0)/(1+s)}$  and estimate for all $u\in B(A)$
  \begin{align*}
    \|A^{-s/2}u\|_{B(A^{1+s})}&=\sum_{j\geq 2}
   R_j^{1/2}\| F_j(A^{1+s})A^{-s/2}u\|\\ & \leq C\sum_{j\geq 2}
   2^{\parb{1-\tfrac{s}{1+s}}(j-1)/2}\bigg\| F\bigg(2^{\tfrac{j-2}{1+s}}\leq
   A<2^{\tfrac{j-1}{1+s}}\bigg)u\bigg\|\\
& \leq C\sum_{j\geq 2}
   p^{(j-1)/2}\big\| F(p^{j-2}\leq
   A<p^{j-1} )u\big\|\\
& = C\|u\|_{B(A)_p}.
  \end{align*} Now we obtain the first estimate in \eqref{eq:4oo} by
  invoking Lemma \ref{lemma:resolvent-boundsoofirst}. 

The second estimate in \eqref{eq:4oo} follows from the first with
$A\to A^{1+s}$ and $s\to -{s}/{(1+s)}$.

\end{proof}

\begin{lemma}
  \label{lemma:resolvent-boundsook} Suppose    $A$ is  a self-adjoint  operator  on a Hilbert
  space $\vH$, $c\in \R$, $u\in B(A)$ and either $|c|>1$
  or $u=F(|cA|\geq 1)u$,  then $u\in B(cA)$ with   \begin{equation}
    \label{eq:4ooB}
    \|u\|_{B(cA)}\leq  8|c|^{1/2} \|u\|_{B(A)}.
  \end{equation}
\end{lemma}
\begin{proof} Suppose  first that $|c|>1$. Pick
  $i\geq 2$ such that $R_{i-1}<|c|\leq R_i$. Then for all $j\geq i+1$
  \begin{equation*}
    F_j(ct)\leq F(R_{j-1}/R_i\leq |t|< R_{j}/R_{i-1})\leq F_{j-i+1}(t)+F_{j-i+2}(t).
  \end{equation*}
Whence for any  $u\in B(A)$ we can estimate
\begin{align*}
  \|u\|_{B(cA)}&\leq  \parb{ \sup_{j\geq
      i+1}\parb{R_{j}/R_{j-i+1}}^{1/2}+\sup_{j\geq
      i+1}\parb{R_{j}/R_{j-i+2}}^{1/2}}\|u\|_{B(A)}+\sum^{i}_{j=1}R_j^{1/2}\|u\|_{\vH}\\
&\leq \parb{
  2^{(i-1)/2}+2^{(i-2)/2}+2^{i/2}(\sqrt 2 +1)}\|u\|_{B(A)}\\
&\leq \parb{\sqrt 2+1+2(\sqrt 2 +1)}|c|^{1/2}
\|u\|_{B(A)}\\
&\leq 8|c|^{1/2}
\|u\|_{B(A)}.
\end{align*}

Suppose now that $|c|\leq 1$ and that $u=F(|cA|\geq 1)u$. Clearly we
can assume that $c\neq 0$. Then we can pick
  $i\geq 2$ such that $R_{i-1}\leq1/|c|< R_i$, and we note that
  $F_1(cA)u=0$. For $j\geq 2$ we have 
\begin{equation*}
    F_j(ct)\leq F(R_{j-1}R_{i-1}\leq |t|< R_{j}R_{i})\leq F_{j+i-2}(t)+F_{j+i-1}(t).
  \end{equation*} Whence 
\begin{align*}
  \|u\|_{B(cA)}&\leq  \parb{ \sup_{j\geq
      2}\parb{R_{j}/R_{j+i-2}}^{1/2}+\sup_{j\geq
      2}\parb{R_{j}/R_{j+i-1}}^{1/2}}\|u\|_{B(A)}\\
&\leq \parb{
  2^{(2-i)/2}+2^{(1-i)/2}}\|u\|_{B(A)}\\
&\leq 3|c|^{1/2}
\|u\|_{B(A)}.\qedhere
\end{align*}
\end{proof}

We note the following abstract version of a result from
\cite{JP, Mo2} (proven by using  suitable decompositions of unity and
the Cauchy-Schwarz inequality, see also \cite[Subsection 2.2]{Wa}).
\begin{lemma}
  \label{lemma:resolvent-boundsooA} Let   $A_1$ and  $A_2$ be self-adjoint
  operators on Hilbert spaces $\vH_1$ and $\vH_2$, 
  respectively, and let $T\in
  \vB(\vH_1,\vH_2)$. Suppose that uniformly  in $ z\in\Gamma_\theta$
and $m,n\in \Z$,
\begin{equation}
  \label{eq:15hA}
   \|F(m\leq A_2< m+1)TF(n\leq A_1< n+1)\|\leq C.
\end{equation} Then $T\in \vB(B(A_1),B(A_2)^*)$, and with the
constant $C$ from \eqref{eq:15hA} we have 
\begin{equation}
  \label{eq:21}
  \|T\|_{\vB(B(A_1),B(A_2)^*)}\leq  2C.
\end{equation}
\end{lemma}
We note the following (partial) abstract criterion  for \eqref{eq:15hA}, cf.  
\cite[(I.10)]{Mo2}  (see also \cite{Wa}). Recall that a bounded
operator $T$ on a Hilbert space is called \emph{accretive} if
$T+T^*\geq 0$, cf.  for example \cite[Chapter X]{RS}.
\begin{lemma}
  \label{lemma:resolvent-boundsooAb} Let   $A$ be self-adjoint
  operator on Hilbert spaces $\vH$,  and suppose  $T\in
  \vB(\vH)$ is accretive.  Suppose 
the following bounds uniformly in $n\in \Z$,
\begin{subequations}
 \begin{align*}
  \|F(n\leq A< n+1)TF(n\leq A< n+1)\|&\leq C_1,\\
\|F(A<
  n)TF(n\leq A< n+1)\|&\leq C_2,\\
\|F(n\leq A< n+1)TF(A\geq n)\|&\leq C_3.
\end{align*} 
\end{subequations} Then \eqref{eq:15hA} holds with $A_1=A_2=A$, the
accretive $T$ and with $C=2C_1+C_2+C_3$.
\end{lemma}

\subsection{Besov
space bound of  resolvent}\label{subsec: Concrete Besov
spaces}

In this subsection we shall prove Theorem~\ref{prop:resolvent-bounds}.
We shall prepare for the proof of \eqref{eq:limitbound2j} in terms of
three lemmas.

Due to \eqref{eq:limitbound},  \eqref{eq:3} and  resolvent identities,
cf. \cite[(5.12)]{FS}, it suffices for \eqref{eq:limitbound2j} to prove
the bound with $V_2$ in Condition \ref{cond:lap}  taken to be zero
(i.e. only Condition \ref{cond:lap} \ref{it:assumption1}--\ref{it:assumption3}
  are imposed and $V=V_1$). Let in
this subsection 
\begin{equation}
  \label{eq:5}
  A=(x\cdot p+p\cdot x)/2;\,p=-\i \nabla_x.
\end{equation} Notice that the commutator
\begin{equation}
  \label{eq:6}
  \i[H,A]=2H +W;\,W(x)=-2V(x)-x\cdot \nabla V(x)\geq \epsilon_1\tilde \epsilon_1\inp{x}^{-\mu}.
\end{equation}

\begin{lemma}
  \label{lemma:resolvent-boundsA} Suppose $V_2=0$ in Theorem
  \ref{thm:lap}. Then with $A$ given by \eqref{eq:5}
\begin{equation}
  \label{eq:limitbound3a}
 \sup_{z \in \Gamma_{\theta}}
\left\|
f_{|z|} R(z)
f_{|z|} \right\|_{\vB (B(A),B(A)^*)} \leq C=C(\theta).
\end{equation}
\end{lemma}
\begin{proof} Following \cite[Section 3]{FS} (a modification of the
  method of \cite{Mo1}) we introduce 
\begin{equation}
  \label{eq:def_R_zeta}
  R_{z}(\epsilon) = (H - \i\epsilon \i[H, A] - z)^{-1};\, \epsilon \,\Im z>0.
\end{equation} We recall the quadratic estimate  \cite[Lemma 3.1]{FS}, 
\begin{align}
\label{eq:quadratic}
\|f_{|z|} R_{z}(\epsilon) T \|^2 \leq C |\epsilon|^{-1} \|T^*
R_{z}(\epsilon) T \|.
\end{align}
valid for  $z\in\Gamma_\theta$ or
$\bar z\in\Gamma_\theta$, $\epsilon\, \Im z>0$  and 
$0<|\epsilon|\leq
\epsilon(\theta)$ sufficiently small and for all bounded operators
$T$. 

We recall, cf.  \cite[Lemma 3.3]{FS}, 
\begin{equation}
  \label{eq:7}
  \frac{d}{d\epsilon} R_{z}(\epsilon) =
(1-2\i\epsilon)^{-1}
\left\{
R_{z}(\epsilon)A -AR_{z}(\epsilon)  +\i\epsilon
R_{z}(\epsilon) (x \cdot \nabla W) R_{z}(\epsilon)
\right\}.
\end{equation}

Also we recall the following bound  valid for any $s\in{\mathbb R}$,
cf.  \cite[(4.9)]{FS},
\begin{equation}
    \label{eq:bounda}
    |\partial_x^{\alpha}f_\lambda^s|\leq C_{\alpha}f_\lambda^s\langle x
\rangle^{-|\alpha|},
  \end{equation}
with $C_{\alpha}$ independent of $\lambda\geq0$.  

Now we shall prove three bounds which are uniform in $
z\in\Gamma_\theta$, $\epsilon>0$ and $0<\nobreak\epsilon\leq\nobreak
\epsilon(\theta)$,
\begin{subequations}
 \begin{align}
  \label{eq:8}
  \|F_z(\epsilon)\|&\leq C \mfor
  F_z(\epsilon):=\inp{A}^{-1}f_{|z|}R_{z}(\epsilon)f_{|z|}\inp{A}^{-1},\\
\label{eq:9}\|F^-_z(\epsilon)\|&\leq C \mfor
  F^-_z(\epsilon):=\e^{\epsilon A}F(A<
  0)f_{|z|}R_{z}(\epsilon)f_{|z|}\inp{A}^{-2},\\
\label{eq:10}\|F^+_z(\epsilon)\|&\leq C \mfor
  F^+_z(\epsilon):=\inp{A}^{-2}f_{|z|}R_{z}(\epsilon)f_{|z|}F(A\geq 0)\e^{-\epsilon A}.
\end{align} 
\end{subequations}

\noindent{\bf Re \eqref{eq:8}.} Due to \eqref{eq:quadratic}
\begin{equation}
  \label{eq:11}
  \|F_z(\epsilon)\|\leq \epsilon^{-1}C\mfor 0<\epsilon\leq \epsilon(\theta).
\end{equation} Next we note the bounds, due to \eqref{eq:bounda},
\begin{equation}
  \label{eq:12}
  \|f_{|z|}^{-1}Af_{|z|}\inp{A}^{-1}\|\leq C\mand \|\inp{A}^{-1}f_{|z|}Af_{|z|}^{-1}\|\leq C.
\end{equation} 
Using \eqref{eq:quadratic}, \eqref{eq:7}  and
\eqref{eq:12} we obtain
\begin{equation}
  \label{eq:13}
  \big \|\frac{d}{d\epsilon} F_{z}(\epsilon) \big \|\leq C\parb{ \epsilon^{-1/2}\|F_z(\epsilon)\|^{1/2}+\|F_z(\epsilon)\|}.
\end{equation} Clearly \eqref{eq:8} follows from \eqref{eq:11} and
\eqref{eq:13} by two integrations.

\noindent{\bf Re \eqref{eq:9}.} Due to \eqref{eq:quadratic} and
\eqref{eq:8}
\begin{equation}
  \label{eq:11-}
  \|F^-_z(\epsilon(\theta))\|\leq \epsilon(\theta)^{-1/2}C.
\end{equation} Using \eqref{eq:7} we compute
\begin{align}
  \label{eq:14}
 \frac{d}{d\epsilon} F^-_{z}(\epsilon)&=T_1+\cdots +T_4;\\ 
T_1&= \parb{1-(1-2\i\epsilon)^{-1}}\e^{\epsilon A}F(A<
  0)Af_{|z|}R_{z}(\epsilon)f_{|z|}\inp{A}^{-2},\nonumber\\
T_2&= (1-2\i\epsilon)^{-1}\e^{\epsilon A}F(A<
  0)[A,f_{|z|}]R_{z}(\epsilon)f_{|z|}\inp{A}^{-2},\nonumber\\
T_3&= (1-2\i\epsilon)^{-1}\e^{\epsilon A}F(A<
  0)f_{|z|}R_{z}(\epsilon)Af_{|z|}\inp{A}^{-2},\nonumber\\
T_4&= \i\epsilon(1-2\i\epsilon)^{-1}\e^{\epsilon A}F(A<
  0)f_{|z|}R_{z}(\epsilon)(x \cdot \nabla W)R_{z}(\epsilon)f_{|z|}\inp{A}^{-2}.\nonumber
\end{align} Using \eqref{eq:quadratic}, \eqref{eq:bounda} and
\eqref{eq:8} we estimate 
\begin{equation}
  \label{eq:11-b}
  \|T_j\|\leq \epsilon^{-1/2}C\mfor 0<\epsilon\leq
  \epsilon(\theta)\mand j=1,\dots,4.
\end{equation}  Notice that for all of the terms $T_1$--$T_4$ we  apply \eqref{eq:quadratic}
with $T=f_{|z|}\inp{A}^{-1}$ and in addition  for $T_4$ we apply \eqref{eq:quadratic}
with $T=f_{|z|}$.  Clearly \eqref{eq:9} follows from
\eqref{eq:11-}--\eqref{eq:11-b} by one integration.

\noindent{\bf Re \eqref{eq:10}.}  We mimic the proof of \eqref{eq:9}.

Next we note that the above arguments apply to $A\to A-n$ for any $n\in \Z$
yielding bounds being independent of $n$. Taking $\epsilon\to 0$ we thus
obtain the following bounds for the
accretive operator $T(z)=-\i f_{|z|}R(z)f_{|z|}$,
all being uniform in $ z\in\Gamma_\theta$ and $n\in \Z$, 
\begin{subequations}
 \begin{align*}
  \|\inp{A-n}^{-1}T(z)\inp{A-n}^{-1}\|&\leq \tilde C,\\
\|F(A<
  n)T(z)\inp{A-n}^{-2}\|&\leq \tilde C,\\
\|\inp{A-n}^{-2}T(z)F(A\geq n)\|&\leq \tilde C.
\end{align*} 
\end{subequations}

Due to these bounds and Lemmas \ref{prop:resolvent-bounds} and
\ref{lemma:resolvent-boundsAbb} we conclude \eqref{eq:limitbound3a}
with $C=\nobreak16 \tilde C$.
\end{proof}

We shall use Weyl quantization on $\R^d$ denoted by $\Opw(c)$, cf.
\cite {Ho2, FS}.  The following ($\lambda$-dependent) symbols will
play a prominant role (cf. \cite {DS3,FS}):
\begin{align}
\label{eq:def_a0_b}
a=a_\lambda= \frac{\xi^2}{f_\lambda(x)^2} ,\;\;
b=b_\lambda=  \frac{\xi}{f_\lambda(x)} \cdot \frac{x}{\langle x \rangle}.
\end{align}

It is convenient to introduce the following metric
\begin{equation*}
 g=g_{\lambda}=\langle
x \rangle^{-2}\d x^2+f_\lambda^{-2}\d\xi^2, 
\end{equation*}
 and the corresponding
symbol classes $S(m, g)$. Here  
 $m=m_{\lambda}=m_{\lambda}(x,\xi)$ will be   a  \emph{uniform weight
 function}, see \cite[Subsection 4.2]{FS} for this terminology and an
 account of basic
 pseudodifferential operator results. Here we have  $\lambda \in [0,\lambda_0]$
 (for the  fixed $\lambda_0>0$)  and  the function $m$ obeying  bounds uniform
 in this parameter, see \cite [Lemma 4.3 (ii)] {FS} for details
 (note however that the
 discussion there concerns uniformity in   a parameter $E\in (0,1]$
 rather than in $\lambda \in [0,\lambda_0]$ which  can be considered
 as a trivial modification). 

In 
 the present  paper we shall use the  notation $S_{\unif}(m_{|z|}, g_{|z|})$, given
 a uniform weight
 function $m$, to signify  the symbol class of smooth symbols
 $c=c_{z}\in  S (m_{|z|}, g_{|z|})$
 with $z$  in the closure of the set $\Gamma_\theta\subset \C$
 satisfying
 \begin{equation}
   \label{eq:sym_class}
  |\partial^\gamma_x \partial^\beta_\xi c_z(x,\xi)|\leq
  C_{\gamma,\beta}m_{|z|}(x,\xi)\langle x \rangle^{-|\gamma|}f_{|z|}^{-|\beta|}.
 \end{equation} Note that these bounds are  uniform  in  $z$ belonging
 to the closure of the set $\Gamma_\theta$.
We also notice that the ``Planck constant'' for this class is 
    $\langle x
    \rangle^{-1}f_{|z|}^{-1}$, in particular $\inp{ x}^{\mu/2-1}$ is
    a ``uniform Planck constant''. 
The corresponding class of Weyl quantized operators  is
    denoted by $\Psi_{\unif}(m_{|z|}, g_{|z|})$. Whence for example
    $f_{|z|}^s\in \Psi_{\unif}(f_{|z|}^s, g_{|z|})$,
    cf. \eqref{eq:bounda}. Here is a list of other examples (referring
    to \eqref{eq:def_a0_b} for definition)
\begin{subequations}
  \begin{align*}
    a_{|z|}&\in S_{\unif}(a_{|z|}+1, g_{|z|}),\\
F(a_{|z|}), I-F(a_{|z|})&\in S_{\unif}(1, g_{|z|})\mforall F\in
C^\infty_c(\R),\\
F(a_{|z|})G(b_{|z|})&\in S_{\unif}(1, g_{|z|})\mforall F,G\in C^\infty_c(\R),\\
h,h-z &\in S_{\unif}(f_{|z|}^2(a_{|z|}+1), g_{|z|})\mfor h:=\xi^2+V(x)=\xi^2+V_1(x).
  \end{align*}
    \end{subequations}

Now,  with reference to the constant $C_0$ in Condition \ref{cond:lap}
 \ref{it:assumption2}
(i.e. the constant with $\alpha=0$) we can bound
\begin{equation}\label{eq:25}
  f_{|z|}^{-2}(x)\big |V(x)-z\big |\leq C_0':=\max(C_0,1).
\end{equation} Let a real-valued $\chi_-\in C^\infty_c(\R)$ be given such that
$\chi_-(t)=1$ for $t\in [0,C_0'+1]$, and let $\chi_+=1-\chi_-$. Then the symbol
\begin{equation*}
  r=r_z:=(h-z)^{-1}\chi_+(a_{|z|})\in S_{\unif}(f_{|z|}^{-2}(a_{|z|}+1)^{-1}, g_{|z|}).
\end{equation*}
In particular, using the calculus (see  \cite[Subsection 4.2]{FS}), we
can write
\begin{subequations}
  \begin{align}
  \label{eq:15}\Opw(\chi_+(a_{|z|}))f_{|z|}&=S_z
  ^{\lef}f_{|z|}^{-1}(H-z)+T_z ^{\lef}\inp{x}^{-1}f_{|z|},\\
\label{eq:16}f_{|z|}\Opw(\chi_+(a_{|z|}))&=(H-z)f_{|z|}^{-1}S_z ^{\righ}+f_{|z|}\inp{x}^{-1}T_z ^{\righ},
  \end{align}
\end{subequations} where  the  operators $S_z
  ^{\lef}, T_z ^{\lef}, S_z ^{\righ}, T_z ^{\righ} $ have symbols
\begin{equation*}
  s_z
  ^{\lef},t_z
  ^{\lef},s_z ^{\righ},t_z ^{\righ}\in S_{\unif}(1, g_{|z|}),
\end{equation*} respectively. Notice for example for \eqref{eq:15} that 
\begin{equation*}
  \Opw(\chi_+(a_{|z|}))-\Opw(r_{|z|})f_{|z|}(H-z)f_{|z|}^{-1}\in\Psi_{\unif}((\inp{x}f_{|z|})^{-1}, g_{|z|}),
\end{equation*} and that this argument can be repeated, say $k$ times,
improving the remainder to be in $S_{\unif}((\inp{x}f_{|z|})^{-k},
g_{|z|})$. Whence if $k(1-\mu/2)\geq 1$ indeed  \eqref{eq:15} follows.
\begin{lemma}
  \label{lemma:resolvent-boundsAb} Under the condition of Lemma
  \ref{lemma:resolvent-boundsA} there exists  $C=C(\theta)>0$
  such that
\begin{equation}
  \label{eq:limitbound3ab}
 \sup_{z \in \Gamma_{\theta}}
\left\|
f_{|z|} R(z)
f_{|z|} \right\|_{\vB (B(f_{|z|}\inp{x}),B(f_{|z|}\inp{x})^*)} \leq C.
\end{equation}
\end{lemma}
\begin{proof}
  Let for convenience $P=f_{|z|} R(z)
f_{|z|} $, i.e. the quantity we want to
bound, and $T_-=\Opw(\chi_-(a_{|z|}))$ (as defined above). By inserting  repeatedly $I=\Opw(\chi_+(a_{|z|}))+\Opw(\chi_-(a_{|z|}))$ either
to the left or right of terms proportional to $P$ and  using
\eqref{eq:15} and \eqref{eq:16} we obtain an expansion in various terms. Only the following four terms are
not obviously in $\Psi_{\unif}(1, g_{|z|})\subseteq \vB(\vH)$ and
hence bounded by  a 
 constant that is  independent of  $z\in
\Gamma_\theta$ (such a constant is 
henceforth referred to as a ``uniform bounding constant''):
\begin{align*}
  P_1&=T_-P\,T_-,\\
P_2&=T_-P \,\inp{x}^{-1}T_z ^{\righ},\\
P_3&=T_z ^{\lef}\inp{x}^{-1}P \,T_-,\\
P_4&=T_z ^{\lef}\inp{x}^{-1}P \,\inp{x}^{-1}T_z ^{\righ}.
\end{align*} Due to \eqref{eq:limitbound} also $P_4\in\vB(\vH)$
with a uniform bounding constant.

We calculate $[A,T_-]\in \Psi_{\unif}(1, g_{|z|})$ and then in turn
\begin{equation*}
  \inp{A}T_-\inp{f_{|z|}\inp{x}}^{-1}\in \vB(\vH),
\end{equation*}  with  a uniform bounding constant.
Whence due to Lemma \ref{lemma:resolvent-bounds} applied with $s=1$
(using for \eqref{eq:18} that $ T_-=T_-^*$)
\begin{subequations}
\begin{align}
  \label{eq:17}
  T_-&\in\vB(B(f_{|z|}\inp{x}),B(A)),\\
T_-&\in\vB(B(A)^*,B(f_{|z|}\inp{x})^*),\label{eq:18}
\end{align} 
\end{subequations} with  uniform bounding constants.

The contribution from the term $P_1$ is treated by \eqref{eq:17}, 
\eqref{eq:18} and Lemma
\ref{lemma:resolvent-boundsA}. 

It remains to treat the contributions from the terms $P_2$ and
$P_3$. For that we use the resolvent identities
\begin{equation}\label{eq:20}
  R(z)=R(\i)+(z-\i) R(z)R(\i)=R(\i)+(z-\i) R(\i)R(z).
\end{equation} We shall treat the term $P_2$ by using the first
identity in \eqref{eq:20},  omitting the
(similar)
arguments for $P_3$ (which is based on the second identity). By Lemma
\ref{lemma:resolvent-bounds} (used again with $s=1$)
\begin{equation} 
  \label{eq:19}
  f_{|z|}^{-1}R(\i)f_{|z|}\inp{x}^{-1}T_z ^{\righ}\in\vB(B(f_{|z|}\inp{x}),B(A)),
\end{equation} with a uniform bounding constant. Now using \eqref{eq:20} there are two terms to bound.  The
contribution from the first term is trivially treated (it is in $\vB(\vH)$), while  the
contribution from the second  term  is treated using \eqref{eq:18}, \eqref{eq:19} and Lemma
\ref{lemma:resolvent-boundsA}. We obtain a uniform bound
$\|P_2\|_{\vB (B(f_{|z|}\inp{x}),B(f_{|z|}\inp{x})^*)} \leq C$. 
\end{proof}

\begin{lemma}
  \label{lemma:resolvent-boundsAbb} There exists  $C=C(\mu)>0$
  such that
\begin{equation}
  \label{eq:limitbound3abb}
\|f_{|z|}^{-1/2}u\|_{B(f_{|z|}\inp{x})}\leq C\|u\|_{B(|x|)}\mforall
z\in\Gamma_\theta \mand u\in B(|x|).
\end{equation}
\end{lemma}
\begin{proof} Due to Corollary \ref{cor:abstr-besov-spac} the bound 
\eqref{eq:limitbound3abb} is equivalent to the bound (possibly
changing the constant) 
\begin{equation}
  \label{eq:limitbound3abbB}
\|f_{|z|}^{-1/2}u\|_{B(f_{|z|}\inp{x})}\leq C\|u\|_{B(\inp{x})}\mforall z\in\Gamma_\theta \mand u\in  B(|x|)=B(\inp{x}).
\end{equation} Now to show \eqref{eq:limitbound3abbB} we decompose
\begin{equation*}
  I=F_-+F_+,\;F_-:=F(|z|<\inp{x}^{-\mu}),\,F_+:=F(|z|\geq\inp{x}^{-\mu}),
\end{equation*} and estimate for all $u\in B(|x|)$
\begin{align*}
\MoveEqLeft[3]  
\|f_{|z|}^{-1/2}F_-u\|_{B(f_{|z|}\inp{x})}\leq \sum_{j=2}^\infty
  R_j^{1/2}\|F(R_{j-2}\leq
  \inp{x}^{1-\mu/2}<R_j)\inp{x}^{\mu/4}F_-u\|\\
\leq{} & \sup_{i\geq 2}  \Big (R_i^{1/2}/ R_{i-1}^{1/2} \Big ) \sum_{j=2}^\infty
  R_{j-1}^{1/2}\|F(R_{j-2}\leq
  \inp{x}^{1-\mu/2}<R_{j-1})\inp{x}^{\mu/4}u\|\\
&+ \sum_{j=2}^\infty
  R_j^{1/2}\|F(R_{j-1}\leq
  \inp{x}^{1-\mu/2}<R_{j})\inp{x}^{\mu/4}u\|.
\end{align*} By  using Lemma \ref{lemma:resolvent-boundsoo}
with $A=\inp{x}$ and  $s=-\mu/2$ to both terms on the right-hand side
we obtain then that
\begin{equation}\label{eq:23}
  \|f_{|z|}^{-1/2}F_-u\|_{B(f_{|z|}\inp{x})}\leq (\sqrt 2+1)C(s)\|u\|_{B(\inp{x})}.
\end{equation}

Similarly we estimate 
\begin{align*}
  \|f_{|z|}^{-1/2}F_+u\|_{B(f_{|z|}\inp{x})}&\leq |z|^{-1/4}\sum_{j=2}^{\infty}
  R_j^{1/2}\|F(R_{j-2}\leq
  \sqrt{|z|}\inp{x}<R_j)F_+u\|\\
&\leq (\sqrt 2+1)|z|^{-1/4}\|F_+u\|_{B(\sqrt{|z|}\inp{x})}.
\end{align*}  
By using Lemma \ref{lemma:resolvent-boundsook}
with $A=\inp{x}$, $c=\sqrt{|z|}$ and $u\to F_+u$ (note that indeed
$F_+u=F(|cA|\geq 1)F_+u$) we obtain in turn that
\begin{equation*}
  \|F_+u\|_{B(\sqrt{|z|}\inp{x})}\leq 8|z|^{1/4}\|F_+u\|_{B(\inp{x})}\leq 8|z|^{1/4}\|u\|_{B(\inp{x})}.
\end{equation*} Consequently 
\begin{equation}\label{eq:24}
  \|f_{|z|}^{-1/2}F_+u\|_{B(f_{|z|}\inp{x})}\leq 8(\sqrt 2+1)\|u\|_{B(\inp{x})}.
\end{equation}

We obtain \eqref{eq:limitbound3abbB}  by combining \eqref{eq:23} and
\eqref{eq:24} (notice that the sum of associated  bounding constants only depends
on $\mu$).
 
\end{proof}
\begin{proof}[Proof of Theorem \ref{prop:resolvent-bounds}.]
  We can assume that $V_2=0$. Due to Lemma
  \ref{lemma:resolvent-boundsAbb}
  \begin{subequations}
    \begin{align}
      \label{eq:17bb}
      f_{|z|}^{-1/2}&\in\vB(B(|x|),B(f_{|z|}\inp{x})),\\
      f_{|z|}^{-1/2}&\in\vB(B(f_{|z|}\inp{x})^*,B(|x|)^*),\label{eq:18bb}
    \end{align}
  \end{subequations} with uniform bounding constants. We decompose
  \begin{equation*}
    f_{|z|}^{1/2}R(z)f_{|z|}^{1/2}= f_{|z|}^{-1/2}\parb{f_{|z|}R(z)f_{|z|}}f_{|z|}^{-1/2},
  \end{equation*} and apply \eqref{eq:17bb}, Lemma
  \ref{lemma:resolvent-boundsAb}  
  and \eqref{eq:18bb} to the three factors on the right-hand side
  ordered from the right to the left, respectively.
\end{proof}

\section{Sommerfeld radiation condition at zero
  energy}\label{subsec: Sommerfeld radiation condition at zero energy}
We shall give a characterization of the boundary values $R(0\pm \i 0)$
of Theorem~\ref{thm:lap} in terms of ``radiation conditions'' in a
Besov space. This may be seen as a zero-energy analogue of the
(strictly) positive-energy result \cite[Theorem 30.2.10]{Ho2} which is
valid for a larger class of potentials. Under less general conditions
than Condition~\ref{cond:lap} such a result was proved in \cite{DS3},
see \cite[Proposition 4.10] {DS3}.

We recall some bounds of \cite{FS}. First it is convenient to change
the definition of the symbols $f,a$ and $b$ used in Subsection \ref{subsec: Concrete Besov
spaces}. Specifically  we  define here and henceforth $f$ by
\eqref{eq:fsublambda} with $K$ given  by \eqref{eq:26}. We define the symbols $a$ and $b$ by
\eqref{eq:def_a0_b} with this value of $K$ and modify similarly the
metric $g$. (The symbol class $S_{\unif}(m_{|z|},
g_{|z|})$ is independent of the constant $K>0$ however.) Let us note the following
modification of \eqref{eq:25}
\begin{equation}\label{eq:25b}
  f_{|z|}^{-2}(x)\big |V_1(x)-z\big |\leq C_0':=\max(C_0/K,1).
\end{equation}  Consider  real-valued $\chi_-\in C^\infty_c(\R)$  such that
$\chi_-(t)=1$  in a neighbourhood of $[0,C_0']$
and
such that $\chi'_-(t)\leq 0$ for $t>0$. Let $\chi_+=1-\chi_-$. Then we obtain
obvious modifications of \eqref{eq:15} and \eqref{eq:16}, in fact we
could replace the power $\inp{x}^{-1}$ occuring to the right in
\eqref{eq:15} and \eqref{eq:16} by any negative power of $\inp{x}$ by
the same arguments used for \eqref{eq:15} and \eqref{eq:16} and we
could distribute powers of $f$ differently. Whence in particular the  pseudodifferential calculus leads
to  the following result.
\begin{lemma}\label{lemma:som-energy}
  Let $\chi_+$ be given as above. Then for  all $n\in \N$ there exist 
  \begin{equation*}
   s_z
  ^{\lef},t_z
  ^{\lef}\in S_{\unif}((a_{|z|}+1)^{-1}, g_{|z|})\subseteq
  S_{\unif}(1, g_{|z|}) 
  \end{equation*}
  such that with $S_z
  ^{\lef}=\Opw(s_z
  ^{\lef})$ and $T_z ^{\lef}=\Opw(t_z
  ^{\lef})$
  \begin{equation}
    \label{eq:27}
   \Opw(\chi_+(a_{|z|})=S_z
  ^{\lef}f_{|z|}^{-2}(H-V_2-z)+f_{|z|}^{-1/2}\inp{x}^{-n}T_z
  ^{\lef}\inp{x}^{-n}f_{|z|}^{1/2}.
  \end{equation}
\end{lemma}

We combine \eqref{eq:limitbound}, \eqref{eq:20}, \eqref{eq:27} and the
calculus to conclude  the
following generalization of \cite[(4.5b)]{FS}. (The method of proof in
\cite{FS} is at this point more complicated.) The continuity assertion
below is shown as in \cite[Subsection
4.1]{DS3}. 
\begin{lemma}
  \label{lemma:high_energy_somm-radi-cond} Let $\chi_-$ and $\chi_+$
  be given as above. Let $\delta>0$ be given as in
  Condition~\ref{cond:lap}~\ref{it:assumption4}. Then for all $s>1/2$,
  $t\in[0,\delta+1-\mu/2] $ there exists $C>0$ such that for $z \in
  \Gamma_{\theta}$
\begin{equation}
  \label{eq:limitboundb}
 \|T_+(z)| \leq C;\;T_+(z):=
\langle x
\rangle^{t-s} f_{|z|}^{1/2 }\Opw(\chi_+(a_{|z|})R(z)
f_{|z|}^{1/2 }\langle x \rangle^{-t-s}.
\end{equation} Moreover there exists $T_+(0+\i0)=\lim_{z \rightarrow 0, z \in \Gamma_{\theta}}
T_+(z)$ in
${\mathcal B}(L^2({\mathbb R}^d))$.
\end{lemma} 
We also have a generalization of  \cite[(4.5d)]{FS}. For that we
consider an arbitrarily given real-valued 
$\tilde \chi_-\in C^\infty_c((-\infty, 1))$. Again we use the $\delta$
from Condition \ref{cond:lap}.
\begin{lemma}
  \label{lemma:low_energy_somm-radi-cond} Let $\chi_-$, $\chi_+$ and $\tilde
  \chi_-$ be given as
  above. Then for all $s>1/2$, $t\in[0,\delta) $  there
  exists $C>0$ such that for $z \in \Gamma_{\theta}$
\begin{equation}
  \label{eq:limitboundbb}
\|
T_-(z)| \leq C;\;T_-(z):=\langle x
\rangle^{t-s} f_{|z|}^{1/2 }\Opw(\chi_-(a_{|z|})\tilde \chi_-(b_{|z|}))R(z)
f_{|z|}^{1/2 }\langle x \rangle^{-t-s}.
\end{equation} Moreover there exists $T_-(0+\i0)=\lim_{z \rightarrow 0, z \in \Gamma_{\theta}}
T_-(z)$ in
${\mathcal B}(L^2({\mathbb R}^d))$.
\end{lemma} 
Note that  the symbol $\chi_-(a_{|z|})\tilde \chi_-(b_{|z|})\in S_{\unif}(1,
g_{|z|})$ while this is not the case for the symbol $\tilde
\chi_-(b_{|z|}$) itself,
and that the important condition on the support of $\tilde \chi_-$
(yielding the localization $b_{|z|}\leq 1-\sigma<1$ for some $\sigma>0$) is
equivalent to the support condition used for the symbol $b$  in
\cite[(4.5d)]{FS} (this is a consequence of   the normalization \eqref{eq:26}). Taking
this remark into account the bound
\eqref{eq:limitboundbb} is essentially contained in \cite{FS}. The
necessary modification of the proof of  the weaker result
\cite[(4.5d)]{FS}  is explained in the beginning of \cite[Subsection
4.1]{DS3}. This explanation will not be repeated  here. Note that a version of
\eqref{eq:limitboundbb} is in fact stated in  \cite{DS3}  as 
\cite[(4.3d)]{DS3}. The continuity assertion of Lemma
\ref{lemma:low_energy_somm-radi-cond} is shown as in \cite[Subsection
4.1]{DS3}. Also we remark that there are versions of Lemmas
\ref{lemma:high_energy_somm-radi-cond} and
\ref{lemma:low_energy_somm-radi-cond} with the localization factors
put to the right of the resolvent, cf. \cite[Theorem 4.1]{FS}, however
these statements will not be used in this paper.

We shall focus on the energy zero for which $f=f_0=\sqrt K\inp
{x}^{-\mu/2}$  and introduce the Besov
space $B^
\mu:=B(\inp{x}^{2s_0})=B(|x|^{2s_0})$ adapted to this energy. Recall
here and henceforth 
the notation of Condition
\ref{cond:lap}, $s_0=1/2+\mu/4$.
  Note that $B^
\mu=f_0^{1/2}B(\inp{x})$, cf. Lemma
\ref{lemma:resolvent-boundsoo}. The corresponding  spaces $(B^
\mu)^*$ and $(B^
\mu)^*_0$ are  characterized as follows, cf. \eqref{eq:2p} and \eqref{eq:41}.
\begin{subequations}
\begin{align}
  &u\in (B^
\mu)^* \Leftrightarrow  u\in L^2_{\rm loc}(\R^d)\mand \sup_{R>1} R^{-s_0}\|F(|x|<R)u\|<\infty,\label{eq:28}\\
&u\in (B^
\mu)^*_0 \Leftrightarrow u\in  L^2_{\rm loc}(\R^d) \mand \lim_{R\to \infty} R^{-s_0}\|F(|x|<R)u\|=0. \label{eq:29}
\end{align} 
In fact the expression to the right in
\eqref{eq:28}  defines a norm on $B(\inp{x}^{2s_0})^*$ which is
equivalent to the canonical norm of \eqref{eq:2}. Yet another
equivalent norm on $(B^
\mu)^* $ is given in terms of an arbitrary  $\epsilon\in (0,1)$ by  the expression $\|F(|x|\leq 1)u\|+\sup_{R>1}
R^{-s_0}\|F(\epsilon R\leq |x|<R)u\|$, and similarly 
\begin{equation}
  \label{eq:42}
  u\in (B^
\mu)^*_0 \Leftrightarrow u\in  L^2_{\rm loc}(\R^d) \mand \lim_{R\to \infty} R^{-s_0}\|F(\epsilon R\leq |x|<R)u\|=0.
\end{equation}
\end{subequations} 

Henceforth
  we  abbreviate $L^2_{s}(|x|)=L^2_{s}$ for any $s\in\R$. The
  corresponding norm is denoted by $\|\cdot\|_s$, and we abbreviate
  $L^2_{0}(|x|)=L^2$ and $\|\cdot\|_0=\|\cdot\|$.
Let $L^2_{-\infty}=\cup_{s\in\R}\, L^2_s$.

Due to Theorem  \ref{thm:lap}, Theorem   \ref{prop:resolvent-bounds}, Lemma
\ref{lemma:high_energy_somm-radi-cond}  and Lemma
 \ref{lemma:low_energy_somm-radi-cond} we have 
\begin{proposition}\label{prop:radiation_somm-radi-conds} Let
  $\chi_-$, $\chi_+$   and $\tilde
  \chi_-$ be given as in Lemma
\ref{lemma:low_energy_somm-radi-cond}. For all $v\in  B^
\mu$ there exists the weak-star limit
\begin{subequations}
\begin{equation}\label{eq:30}
  u=R(0 + \i0) v
= \wslim_{z \rightarrow 0, z \in \Gamma_{\theta}}
R(z) v\in (B^
\mu)^*,
\end{equation} and
\begin{equation}\label{eq:31}
  \Opw(\chi_+(a_0))u,\;\Opw(\chi_-(a_0)\tilde \chi_-(b_0))u\in (B^
\mu)^*_0.
\end{equation} 
  \end{subequations}

If $v\in  L^2_{s}$ for some $s>s_0$ we have the following
stronger conclusion compared to \eqref{eq:31},
\begin{equation}\label{eq:31ab}
  \Opw(\chi_+(a_0))u,\;\Opw(\chi_-(a_0)\tilde \chi_-(b_0))u\in L^2_{t} \text{ for all }t < \min(s-s_0,\delta)-s_0.
\end{equation} In particular for any such $v$ we can take
$t=-s_0$ in \eqref{eq:31ab}.
\end{proposition}
\begin{proof} Note that indeed \eqref{eq:30} follows from Theorems
  \ref{thm:lap} and  \ref{prop:resolvent-bounds} (by a
  density argument). Similarly it suffices for \eqref{eq:31} to show
  the statements for $v\in L^2_1$ (using here that
  $\Opw(\chi_+(a_0))$ and $\Opw(\chi_-(a_0)\tilde \chi_-(b_0))$ are   bounded operators on $(B^
\mu)^*$, cf. Lemma \ref{lemma:resolvent-bounds}). However for such $v$ it follows from  Lemmas
\ref{lemma:high_energy_somm-radi-cond} and 
\ref{lemma:low_energy_somm-radi-cond} that
\begin{equation*}
  \Opw(\chi_+(a_0))u,\;\Opw(\chi_-(a_0)\tilde \chi_-(b_0))u\in  L^2_{-s_0}\subseteq (B^
\mu)^*_0.
\end{equation*} Clearly \eqref{eq:31ab} is also a consequence of
Lemmas \ref{lemma:high_energy_somm-radi-cond} and 
\ref{lemma:low_energy_somm-radi-cond}.
\end{proof}

Note that $s_0$ defines the critical weighted $L^2$ space for the
function $u$ defined in \eqref{eq:30}, more precisely $u\in L^2_s$ for
all $s<-s_0$ while no stronger assertion of this type is given.
Whence this $u\in L^2_{-s_0-\delta}$ and it constitutes a particular
(distributional) solution to the equation $Hu=v$. We study this
situation somewhat generally in the following version of the
``propagation of singularities'' result \cite[Proposition~4.5]{DS3},
see also \cite{Ho3,Me,Va}.
\begin{proposition}
  \label{prop:propa_sing} Let $\chi_-$ be given as in Lemma
\ref{lemma:low_energy_somm-radi-cond}. 
\begin{enumerate}[i)]
\item \label{item:1}
Suppose  $v\in  B^
\mu$, that  $u\in L^2_{-s_0-\delta}$  obeys $Hu=v$,  and the following
localization for some $\sigma >0$
\begin{subequations}
\begin{equation}\label{eq:31bA}
  \Opw(\chi_-(a_0)\tilde \chi_-(b_0))u\in L^2_{-s_0}\mforall \tilde \chi_-\in C^\infty_c((-\infty, \sigma-1)).
\end{equation} Then 
\begin{equation}\label{eq:31bB}
  \Opw(\chi_-(a_0)\tilde \chi_-(b_0))u\in L^2_{-s_0}\mforall \tilde \chi_-\in C^\infty_c((-\infty, 1)).
\end{equation}
  \end{subequations}
  \begin{subequations}
\item \label{item:2}Suppose  $v\in  L^2_{s+2s_0}$ for some $s\in\R$, 
  that $u\in L^2_{s-\delta}$  obeys $Hu=v$,  and the following  localization for
  some $\sigma >0$
\begin{equation}\label{eq:31bC}
  \Opw(\chi_-(a_0)\tilde \chi_-(b_0))u\in L^2_{s} \mforall \tilde \chi_-\in
  C^\infty_c((-\infty, \sigma-1)).
\end{equation} Then 
\begin{equation}\label{eq:31bBC}
  \Opw(\chi_-(a_0)\tilde \chi_-(b_0))u\in L^2_{s}\mforall \tilde \chi_-\in C^\infty_c((-\infty, 1)).
\end{equation}
   \end{subequations}
\end{enumerate} 
\end{proposition}
\begin{proof} Obviously \ref{item:1} follows from \ref{item:2} with
  $s=-s_0$. We shall show \ref{item:2} by  mimicking  the scheme of the proof of \cite[Proposition 4.5]{DS3}. Since the
  observable $b$ used in \cite{DS3}  is somewhat different to our
  $b_0$ and the class of 
  potentials   considered  is smaller than here, there are indeed
  various differences to tackle. We intend to give a
  self-contained presentation.

We introduce the notation $s_1=2s_0-\mu+\delta$, 
  and explain the role of this  parameter: Suppose $u\in
L^2_{s-s_1}$  obeys $Hu=v\in  L^2_{s+2s_0}$. Note that since 
  $s_1>\delta$ this is a more general situation than in \ref{item:2}. Then due
  to Lemma \ref{lemma:som-energy} we can find  $S\in
  \vB(L^2)$   
    such that $X^{s}SX^{-s}\in
  \vB(L^2)$  and 
\begin{equation*}
  \Opw(\chi_+(a))u-SX^{-s_1}u\in  L^2_{s},
\end{equation*} and whence we conclude that  
\begin{equation}\label{eq:35AA}
  \Opw(\chi_+(a))u\in  L^2_{s}.
\end{equation}  A consequence of this remark is that the statements
\eqref{eq:31bC} and \eqref{eq:31bBC} do not depend on which
$\chi_-=1-\chi_+$ that enter (in particular the factors of $\chi_-$ in these
statements may differ from one another).

At various points below we need to bound possible local singularities of the
potential $V_2$. Since these by assumption are located in  a bounded
region they are easily  treated by the general bounds (in
turn easily proven using the relative compactness of  $V_2$)
\begin{equation}
  \label{eq:40}
  \|V_2F(|x|\leq R)u\|\leq C_t(\|Hu\|_t+\|u\|_t) \mfor t\in \R.
\end{equation} The reason why we need  the a priory  information
$u\in L^2_{s-\delta}$ in \ref{item:2}  rather than just $u\in
L^2_{-\infty}$ (as  in \cite{DS3}) lies in
contributions from $V_2F(|x|> R)u$. We will use that the function 
$|x|^{2s_0+\delta}V_2F(|x|> R)$ is bounded (for $R$ as in Condition
\ref{cond:lap} \ref{it:assumption4}), but since we are not
imposing further regularity of this function  our
method of proof does not allow us to weaken the condition $u\in L^2_{s-\delta}$.
 
Henceforth we  abbreviate the symbols $a_0=a$ and  $b_0=b$.
An important ingredient of the proof  is the following estimation of
a Poisson bracket, cf. \cite[(4.30)]{DS3}. Here we have
$h_1=\xi^2+V_1(x)$.
\begin{align}\label{eq:32}
\begin{split}
  \{h_1,b\}&=\inp{x}^{\mu/2-1}\parb{W-(2-\mu)
    K\inp{x}^{-\mu}b^2  +2h_1}/\sqrt K\\
&\geq  (2-\mu)\sqrt K\inp{x}^{-2s_0}\parb{1-b^2}+\inp{x}^{\mu/2-1}2h_1/\sqrt K.
\end{split}
\end{align}

We introduce 
for $\kappa\in (0,1]$ the notation $X_\kappa=(\kappa|x|^2+1)^{1/2}$ and
abbreviate $X_1=X=\inp{x}$. These
observables have 
Poisson brackets 
\begin{equation}
  \label{eq:33}
  \{h_1,X_\kappa\}=2\kappa\xi\cdot x/X_\kappa=2f_0b\tfrac{\kappa X}{X_\kappa}\mand \{h_1,X\}=2\xi\cdot x/X=2f_0b.
\end{equation}
 
\medskip

\noindent\emph{Step I.} Let 
  \begin{equation}
    \label{eq:34}
     \epsilon_2=\min (1-\mu/2,\delta).
  \end{equation}  The quantity  $X^{-\epsilon_2}$ will play  the role of a
``Planck constant''. (The number  $ 2\epsilon_2$ will play the
  role of  the  parameter  $\epsilon_2$ used in \cite{DS3}.) First we prove   \eqref{eq:31bBC}  with the
replacement 
$s\to \tilde s=s+\epsilon_2 -\delta$ for any $u\in L^2_{\tilde s- \epsilon_2}=
 L^2_{s-\delta}$ obeying
 $Hu=v\in  L^2_{\tilde s+2s_0} \supseteq L^2_{ s+2s_0}$. Note that
 indeed 
  $\tilde s\leq s$.   In particular if \eqref{eq:31bC}  is valid for
  $s$  then \eqref{eq:31bC}   is also valid for $s\to \tilde s$ (to be used
  in Step II below). So we assume in addition \eqref{eq:31bC} with
  $s\to \tilde s$ for any such $u$, and we aim at proving
  \eqref{eq:31bBC} with $s\to \tilde s$. Clearly  we can (and will) assume
  that 
  $\sigma<1$ in \eqref{eq:31bC}. 

We shall use that 
\begin{equation}\label{eq:35}
  \Opw(\chi_+(a))u\in  L^2_{\tilde s},
\end{equation}  cf. \eqref{eq:35AA}.

  Now let $\tilde \chi_-$ be given as in \eqref{eq:31bBC}. In particular this means
  that $\supp \tilde \chi_-\subseteq (-\infty, k]$ for some
  $k\in(0,1)$. Pick a non-positive  function $f\in
  C_c^{\infty}((\sigma/3-1,1))$ with $f'(t)\geq0$
  on $[\sigma/2-1,\infty)$ and $f(t)<0$ on
  $[\sigma/2-1,(k+1)/2]$. Pick real-valued  functions $f_1, f_2\in
  C_c^{\infty}((-1,1))$ with $f_1^2(t)+f_2^2(t)=1$ on $\supp f$ and $\supp
  f_1\subset (\sigma/3-1,\sigma-1)$ while $\supp
  f_2\subset (\sigma/2-1,1)$. We introduce for any $K_0>0$
  and $\kappa\in (0,1]$ the
  observables  
\begin{equation} 
\label{eq:propobsi}
  b_\kappa=X^{s_0}a_\kappa,\; a_\kappa=X^{\tilde s}X_{\kappa}^{-\epsilon_2}\exp({-K_0b})f ( b)\chi_-(a);
 \end{equation} here $X_\kappa$ is defined above. 
Notice that we removed the
 subscripts for the old observables $a$ and $b$ and used these in
 \eqref{eq:propobsi} to
 introduce new  observables  with subscripts 
 (not to be mixed up with the old notation). We will prove bounds involving quantizations of
 $a_\kappa$ and $b_\kappa$  that will be uniform in $\kappa$. The
 proof will  then be completed by taking $\kappa\to 0$.

  We look at the right hand side of \eqref{eq:32}. The   first term has the following
  positive lower bound on $\supp b_\kappa$: 
 \begin{equation*}\cdots \geq cX^{-2s_0};\;c=(2-\mu)\sqrt{K} \big(1-\sup \{t^2|t\in \supp f\}\big ).\end{equation*}

   First we fix $K_0$: A part of the Poisson bracket with
 $b_\kappa^2$ is
\begin{equation}
\label{eq:Poi1i}
  \{h_1,X^{2\tilde s+2s_0}X_{\kappa}^{-2\epsilon_2}\}=
  Y_\kappa  bX^{2\tilde s}X_{\kappa}^{-2\epsilon_2},
\end{equation}
where $Y_\kappa=Y_\kappa(|x|)$ is uniformly bounded in $\kappa$, cf. \eqref{eq:33}.
We pick $K_0>0$ such that for all $\kappa\in (0,1]$ \[2K_0c\geq |Y_\kappa  |+2\text{ on }
\supp b_\kappa.\]

From \eqref{eq:32}, \eqref{eq:Poi1i} and the properties of $K_0$ and $f$,  we conclude the following bound at  $\{ f'(b)\geq0\}$: 
\begin{equation*}
  \{h_1,b_\kappa^2\}\leq -2a_\kappa^2 +O\big
  (X^{\mu+\tilde s}\big )h_1a_\kappa  +O\big (X^{2\tilde s}\big )(\chi_-^2)'(a).
\end{equation*}
   Next we multiply both sides  by
  $f_2^2(b)\,(=1-f_1^2(b))$ and obtain after a rearrangement 
  \begin{align}\label{eq:Poi3iyy}
    \begin{split}
       \{h_1,b_\kappa^2\}\leq{} & -2a_\kappa^2 +h_1X^{\mu} d_{\kappa}a_\kappa 
  \\
  &+K_1f_1^2(b)\chi^2_-(a)X^{2\tilde s}-K_2(\chi_-^2)'(a)X^{2\tilde s},\; d_{\kappa}\in S(X^{\tilde s}, 
  g_{0});
    \end{split}
\end{align} here $K_1,K_2>0$ are independent of $\kappa$, and the
family of 
symbols $\{d_{\kappa}:\kappa\in (0,1]\}$ is  bounded in the specified
symbol class.

We introduce $A_\kappa={\Opw}(a _\kappa)$, $B_\kappa={\Opw}(b
_\kappa)$ and the regularization 
\[ 
u_R=\chi(X/R<\nobreak 1)u
\]
in terms of a parameter $R>1$. Here and henceforth we use the notation
$\chi(t>\epsilon)$ for any $\epsilon>0$ to denote a smooth increasing
function $=1$ for $t>\epsilon$ and $=0$ for $t<\frac {1}{2}\epsilon$
and we define $\chi(\cdot<\epsilon)=1-\chi(\cdot>\epsilon)$. Let
$H_1=H-V_2$. The following arguments rely heavily on the calculus, cf.
\cite[Theorems 18.5.4, 18.6.3, 18.6.8]{Ho2}.

First we compute the expectation
\begin{equation}
  \label{eq:comm1i}
 \langle \i[H_1, B_\kappa^2]\rangle _u= \lim_{R\to\infty}\langle \i[H_1,
 B_\kappa^2]\rangle _{u_R}=-2\Im \langle v,B_\kappa^2 u\rangle+2\Im \langle V_2u,B_\kappa^2 u\rangle.  
\end{equation} Next we  estimate 
\begin{equation}
  \label{eq:38}
  |-2\Im \langle v,B_\kappa^2 u\rangle|\leq  C_1\|v\|_{\tilde s+2s_0}\big (\|A_\kappa
 u\|+\|u\|_{\tilde s-\epsilon_2}\big)
\leq \tfrac{1}{4} \|A_\kappa
 u\|^2 +C_2,
\end{equation} and (using  the estimate
\eqref{eq:40} with any $t\leq \tilde s-\delta$)
\begin{align}
\label{eq:39}
\begin{split}
  |2\Im \langle V_2u,B_\kappa^2 u\rangle|&\leq  C_3\big
(\|u\|_{\tilde s-\delta}+\|v\|_{\tilde s+2s_0}\big )\big (\|A_\kappa
 u\|+\|u\|_{\tilde s-\epsilon_2}\big)\\
&\leq \tfrac{1}{4} \|A_\kappa
 u\|^2 +C_4.
\end{split}
\end{align}
 From  \eqref{eq:comm1i}--\eqref{eq:39} we conclude  that 
\begin{equation}
  \label{eq:comm2i}
 |\langle \i[H_1, B_\kappa^2]\rangle _u|\leq \tfrac{1}{2} \|A_\kappa
 u\|^2 +C_2+C_4.
\end{equation}

 On the other hand, using  \eqref{eq:31bC} (with $s\to \tilde s$), \eqref{eq:35}  and 
 \eqref{eq:Poi3iyy},   we infer that 
\begin{align*}
  \langle \i[H_1, B_\kappa^2]\rangle _u&=\lim_{R\to\infty}\langle
  \i[H_1, B_\kappa^2]\rangle _{u_R}\\&\leq -2\|A_\kappa
  u\|^2+C_5\|H_1u\|_{\tilde s+\mu}\|A_\kappa
 u\|+C_6.
\end{align*} Here the second term arises  as a bound of $\Re\inp{{\Opw}(d
_\kappa)X^\mu H_1u,A_\kappa u}$. Using the bound $\|(H-V_2)u\|_{\tilde
  s+\mu}\leq C\big (\|v\|_{\tilde s+\mu}+\|u\|_{\tilde s-\delta}\big)$
it follows that 
\begin{equation}
  \label{eq:comm3i}
 \langle \i[H_1, B_\kappa^2]\rangle _u\leq -\tfrac{3}{2} \|A_\kappa
 u\|^2+ C_7.
\end{equation} 

Combining \eqref{eq:comm2i} and \eqref{eq:comm3i} yields 
\begin{equation*}
 \|A_\kappa u\|^2\leq C_2+C_4+C_7,
\end{equation*} 
which in combination with  the property,  $f(t)<0$ on
  $[\sigma/2-1,(k+1)/2]$,   in
turn gives a uniform bound 
\begin{equation}
  \label{eq:comm4i}
 \|X_{\kappa}^{-\epsilon _2}\Opw(\chi_-(a)\tilde \chi_-(b))u\|^2_{\tilde s}\leq C.
\end{equation}

We let $\kappa\to 0$ in \eqref{eq:comm4i} and infer \eqref{eq:31bBC}
with $s\to \tilde s$.

\medskip

\noindent\emph{Step II.} Define for $m\in \N$ the number $\tilde
s_m=\min (s, s+m\epsilon_2-\delta )$, where $\epsilon_2$ is given by
\eqref{eq:34}.  We learn from Step I that \eqref{eq:31bBC} is valid
with $s\to \tilde s_1$.  We proceed inductively to prove that
\eqref{eq:31bBC} is valid with $s\to \tilde s =\tilde s_m$ given that
we know the statement for $\tilde s_{m-1}$ for some
$m=2,3,\dots$ This is done by mimicking the corresponding step in
\cite{DS3} although in the present context it is doable in a somewhat
simpler way. The idea is to use the procedure of Step I for a
localized version of $u$, say $u_\epsilon=I_\epsilon u$. The factor
$I_\epsilon$ is a pseudodifferential operator with symbol $=1$ in a
slightly bigger region than the support of the symbol $a_\kappa$ of
Step I (now used with $\tilde s =\tilde s_m$). Explicitly $I_\epsilon$
can be constructed as follows: With $f$ and $\chi_-$ given as in
\eqref{eq:propobsi} we pick $f_\epsilon\in C^\infty_c((-\infty, 1))$
with $f_\epsilon(t)=1$ in a neighbourhood of the support of $f$ and we
pick $\chi_\epsilon$ of the same type of function as $\chi_-$ but with
$\chi_\epsilon(t)=1$ in a neighbourhood of the support of $\chi_- $,
and then we define
$I_\epsilon=\Opw(\chi_\epsilon(a)f_\epsilon(b))$. By the induction
hypothesis $u_\epsilon\in L^2_{\tilde s_{m-1}}$. We need to consider
contributions from the commutator $[H_1,I_\epsilon]$ when mimicking
the procedure of Step I using the same constructions
\eqref{eq:propobsi}. But these are ``nice to any order'' due to the
support properties of the involved symbols and the calculus. To see
this more concretely consider the analogue of \eqref{eq:38}: We
consider now
\begin{equation*}
  \langle I_\epsilon  v+[H_1,I_\epsilon]u,B_\kappa^2  u_\epsilon\rangle=\langle I_\epsilon  v,B_\kappa^2  u_\epsilon\rangle+\langle B_\kappa^2 [H_1,I_\epsilon]u, u_\epsilon\rangle.
\end{equation*}
The contribution from the first term is treated as before (using now
the induction hypothesis), and the contribution from the second term
is indeed harmless (since the symbol of $B_\kappa^2[H_1,I_\epsilon]$
is in $S(X^t,g_0)$ for all $t\in \R$). Similarly the expression
$\inp{V_2u,B_\kappa^2 u}$ of \eqref{eq:39} needs to be replaced by
\begin{equation*}
  \langle  I_\epsilon V_2u,B_\kappa^2  u_\epsilon\rangle,
\end{equation*} which indeed can be estimated as before 
(using  the
induction hypothesis). Whence we obtain \eqref{eq:comm2i} with $u\to
u_\epsilon$. 

Similarly we can show \eqref{eq:comm3i}, and hence in turn conclude \eqref{eq:comm4i},
for $u\to
u_\epsilon$. Again we complete the proof by letting $\kappa\to 0$.
\end{proof}

 We have two versions of uniqueness of the outgoing solution at zero energy stated in one theorem as follows.
\begin{thm}
  \label{thm:somm-radi-conds} Let $\chi_-$ be given as in  Lemma
\ref{lemma:low_energy_somm-radi-cond}.  Let $\delta$ and $s_0$  be
given as in Condition \ref{cond:lap}. 
\begin{enumerate}[i)]
  \begin{subequations}
\item\label{item:3} Let $v\in  B^
\mu$. Then the equation $Hu=v$ (in the
  distributional sense)   has a unique solution $u\in
L^2_{-\infty}$ obeying  
\begin{equation}\label{eq:31bddd}
 \exists \;\sigma\in (0,1):\;\Opw(\chi_-(a_0)\tilde \chi_-(b_0))u\in (B^
\mu)^*_0 \mforall \tilde \chi_-\in C^\infty_c((-\infty, \sigma)).
\end{equation} 
this  solution is given by  $u=R(0 + \i0) v$ as defined in
\eqref{eq:30}. In particular    
\begin{equation}\label{eq:31b}
  \Opw(\chi_-(a_0)\tilde \chi_-(b_0))u\in (B^
\mu)^*_0 \mforall \tilde \chi_-\in C^\infty_c((-\infty, 1)).
\end{equation}
   \end{subequations} 
   \begin{subequations}
\item \label{item:4} Let $v\in  L^2_{s}$ for some $s>s_0$.  Then the equation $Hu=v$  has a unique solution $u\in
L^2_{-s_0-\delta}$ obeying 
\begin{equation}\label{eq:31bAbb}
  \exists \;\sigma>0:\;\Opw(\chi_-(a_0)\tilde \chi_-(b_0))u\in L^2_{-s_0}\mforall \tilde \chi_-\in C^\infty_c((-\infty, \sigma-1)).
\end{equation} 
This solution is given by $u=R(0 + \i0) v$ as defined in
\eqref{eq:30}. Whence  for all  $t < \min(s-s_0,\delta)-s_0$
\begin{equation}\label{eq:31bAbbb}
 \Opw(\chi_-(a_0)\tilde \chi_-(b_0))u\in L^2_{t}\mforall \tilde
 \chi_-\in C^\infty_c((-\infty, 1)).
\end{equation}  In particular we can take $t=-s_0$ in
\eqref{eq:31bAbbb}.
   \end{subequations}
\end{enumerate}

\end{thm}
\begin{proof}  By Proposition \ref{prop:radiation_somm-radi-conds}
the
function   $u=R(0 + \i0) v$ defined as in  \eqref{eq:30}  for any $v$
given as in \ref{item:3} or \ref{item:4} obeys
\eqref{eq:31b} or \eqref{eq:31bAbbb} (for any such $t$), respectively. By Proposition
\ref{prop:propa_sing} \ref{item:1} the condition \eqref{eq:31bAbb} is
stronger than \eqref{eq:31b}, which in turn obviously is stronger than
\eqref{eq:31bddd}. Whence it only remains to show that
\eqref{eq:31bddd} is sufficient for uniqueness. In fact we can take
$v=0$, and it  suffices to show that  \eqref{eq:31bddd} for some $u\in
L^2_{-\infty}$ obeying $Hu=0$ implies that $u=0$.
For that we  essentially mimic the proof of \cite[Proposition
4.10]{DS3}.

\medskip
\noindent\emph{ Step I.} We shall show that any such $u\in L^2_{s}$ for
all $s<-s_0$. Suppose $u\in
L^2_{t}$ for some $t<{-s_0}$ (which is an a priory
information). Introducing 
\begin{align}
\epsilon=-\min\big ((t+s_0)/2,\;t+s_0+\epsilon_2\big )\mand\nonumber
t_1=-s_0-\epsilon \nonumber
  \end{align}  (the parameter $\epsilon_2$ is given in  \eqref{eq:34}),
 we check that
\begin{equation}\label{eq:43}
 t_1>t\geq  t_1 -\epsilon_2.
\end{equation} Let  $u_R=\chi(X/R<1)u$, $R>1$,  be a given  regularization of
$u$ as in the proof of
Proposition \ref{prop:propa_sing}. We abbreviate
$\chi_R=\chi(X/R)=\chi(X/R<1)$. Here and henceforth we also use the notation of
\eqref{eq:33} (in particular we again abbreviate $a_0=a$ and  $b_0=b$). Obviously, undoing the  commutator, we have
\begin{equation}
  \label{eq:37}
 \inp{\i [H,X^{-2\epsilon}\chi_R]}_u=0.
\end{equation}

On the other hand, cf. \eqref{eq:33},  
\begin{align*}
     \label{eq:uncomm3}
   \i [H, X^{-2\epsilon}\chi_R]&=2\Re \big (f_0 h_{\epsilon,R}
   \Opw(b)\big );\\h_{\epsilon,R}(x)&=-2\epsilon
   X^{-1-2\epsilon}\chi(X/R)+ X^{-2\epsilon}R^{-1}\chi' (X/R ).
 \end{align*} By using Lemma \ref{lemma:som-energy}, \eqref{eq:40},
 \eqref{eq:43} and the calculus (cf. \cite[Theorems~18.5.4 and~18.6.3]{Ho2}) we obtain that
\begin{align*}
  2\Re \langle f_0 h_{\epsilon,R}
   \Opw\big (b\chi_+(a)\big)\rangle_{u} =2\Re \langle f_0 h_{\epsilon,R}
   \Opw\big (b\chi_-(a)\chi_+(a)\big)\rangle_{u} +O(R^0)=O(R^0),
 \end{align*}  and whence   that 
\begin{equation*}
   \langle \i [H, X^{-2\epsilon}\chi_R]\rangle_u = 2\Re \langle f_0 h_{\epsilon,R}
   \Opw\big (b\chi_-^2(a)\big)\rangle_{u} +O(R^0).
   \end{equation*} We invoke  then \eqref{eq:31bddd} and  deduce
\begin{equation*}
     \label{eq:uncomm8}
   \langle \i [H, X^{-2\epsilon}\chi_R]\rangle_u = 2\Re \langle f_0 h_{\epsilon,R}
   \Opw\big (b\chi(b>\sigma/2)\chi_-^2(a)\big)\rangle_{u} +O(R^0),
   \end{equation*} 
  which in turn  yields (by applying \cite[Theorem
   18.6.8]{Ho2} and ``reversing'' the arguments above) that 
\begin{equation}
     \label{eq:uncomm9}
     \langle \i [H, X^{-2\epsilon}\chi_R]\rangle_u \leq -\epsilon
     \sigma \langle f_0 X^{-1-2\epsilon}\chi_R \rangle_{u} +O(R^0).
   \end{equation} 

By combining  \eqref{eq:37} and \eqref{eq:uncomm9} we obtain 
\begin{equation}
     \label{eq:uncomm10}
     \langle f_0X^{-1-2\epsilon}\chi_R \rangle_u \leq C,
   \end{equation} for some constant $C$ which is independent of $R>1$. Whence,
   letting $R\to \infty$ we see that $u\in L^{2}_
   {t_{1}}$.  

More generally, we define for $m\in\N$ \[t_m=-s_0+\min \big((t_{m-1}+s_0)/2,\;
  t_{m-1}+s_0+\epsilon_2\big ),\;t_0:=t,\]and  iterate the above procedure. We
  conclude that  $u\in 
  L^2_{t_m}$, and hence  that indeed $u\in
  L^2_{s}$ for all  $s<-s_0$.

\medskip

\noindent\emph{Step II.} We shall show that $u\in (B^\mu)^*_0$. We apply
the same scheme as in  Step I, now with $\epsilon=0$ and using the
same  factor of 
$\chi(b>\sigma/2)$. This leads to 
\begin{equation*}
     -R^{-1}\langle f_0 \chi'\big (X/R\big )\rangle_u =o(R^0),
   \end{equation*} and hence indeed   
   $u\in (B^\mu)^*_0$, cf. \eqref{eq:42}.

\medskip

\noindent\emph{Step III.} We shall show that $u=0$. First  we introduce
$\epsilon_3=\epsilon_2/2$ and 
let
$s\in (-s_0,\epsilon_3 -s_0)$  be given arbitrarily.   Our goal  is  to 
show that $u\in L^2_s$. An iteration procedure will then give that
$u\in L^2$, and hence that $u=0$. 

We consider for  $\kappa\in (0,1/2]$ (and in terms of notation of
\eqref{eq:33}) 
\begin{equation} 
\label{eq:propobs8}
  b_\kappa=X^{s_0}a_\kappa;\; a_\kappa=\Big (\tfrac {X}{X_\kappa}\Big )^{s}X_{\kappa}^{-s_0}\chi(-b>1/2)\chi_-(a).
\end{equation} 
Using   \eqref{eq:33} we calculate the Poisson bracket 
\begin{equation*}
   \Big\{h_1, \Big(\tfrac {X}{X_\kappa}\Big )^{2s_0+2s} \Big\}=4(1-\kappa)
   (s_0+s)X_\kappa^{-3}\Big (\tfrac {X}{X_\kappa}\Big )^{2s_0+2s-1}f_0b. 
\end{equation*} This is  negative  on the support of
   $b_\kappa$   with the (negative)  upper bound
   \begin{align}\cdots &\leq -2^{-1}(s_0+s)\Big (X^{2s_0-1}f_0\Big
     )X_\kappa^{-2}\Big (\big (\tfrac {X}{X_\kappa}\big
     )^{s}X_{\kappa}^{-s_0}\Big)^2 \nonumber\\&= -cX_\kappa^{-2}\Big (\big
     (\tfrac {X}{X_\kappa}\big
     )^{s}X_{\kappa}^{-s_0}\Big)^2;\;c=2^{-1}(s_0+s)/\sqrt K.\label{eq:44}\end{align} 

Similarly, by  \eqref{eq:32}, 
\begin{align*}
   \MoveEqLeft[1]
   \big\{h_1, \chi(-b>1/2)\big\}
   \\&\leq-(2-\mu)\sqrt K\chi'(-b>1/2)X^{-2s_0}(1-b^2)-2\chi'(-b>1/2)X^{\mu/2-1}h_1/\sqrt K. 
\end{align*} Note  that the  first term is non-positive.

We introduce the quantizations $A_\kappa={\Opw}(a
_\kappa)$ and $B_\kappa={\Opw}(b
_\kappa)$, and  the states $u_R=\chi_Ru$, $R>1$. Since $u\in
(B^\mu)^*_0$ due to   Step II
\begin{equation}\label{eq:0bn}
  \lim _{R\to \infty} \langle \i [H, B_\kappa^2]\rangle_{u_R}=0.
\end{equation} Since $\delta> 2\epsilon_3$ and $s-\epsilon_3<-s_0$ we  also that
\begin{equation*}
  |\langle \i [V_2, B_\kappa^2]\rangle_{u_R}|\leq C,
\end{equation*} where $C$ is a positive constant  independent of
$R>1$ and $\kappa\in (0,1/2]$.

On the other hand  due to the above considerations the expectation of
$\i [H_1, B_\kappa^2]$ in $u_R$ tends to be negative. Using  Lemma
\ref{lemma:som-energy}, \eqref{eq:40}, \eqref{eq:0bn}  and the above estimations of
symbols we obtain by letting  
$R\to \infty$ 
\begin{equation*}
  c\|X_\kappa^{-1}A_\kappa u\|^2\;\big (=\lim_{R\to \infty} c\|X_\kappa^{-1}A_\kappa u_R\|^2\big )\leq C,
\end{equation*}  where $c$ is given by \eqref{eq:44} and   $C$ is independent of
$\kappa$. 

Whence, letting $\kappa\to 0$, we conclude that 
\begin{subequations}
\begin{equation}\label{eq:bor1}
  \Opw\Big( \chi(-b>1/2)\chi_-(a)\Big)u\in L^{2}_
  {s}.
\end{equation}

Upon replacing the factor $\chi(-b>1/2)$ in \eqref{eq:propobs8} by
$\chi(b>1/2)$, we can argue similarly and obtain
\begin{equation}\label{eq:bor2}
  \Opw\Big( \chi(b>1/2)\chi_-(a)\Big)u\in L^{2}_
  {s}.
\end{equation} 
\end{subequations}

In combination with Lemma \ref{lemma:som-energy}, \eqref{eq:40} and Proposition
\ref{prop:propa_sing}   the bounds
\eqref{eq:bor1} and \eqref{eq:bor2} yield that $u\in L^{2}_{
  s}$.

Next  the above procedure is  iterated: Assuming that $u\in L^{2}_{
  s}$ for all $s<t_m:= m\epsilon_3-s_0$ (for
  some 
  $m\in \N$), it  leads to $u\in L^{2}_{s}$ for all $s<t_{m+1}$.
  Consequently, 
  $u\in L^{2}_{s}$ for all $s\in \R$. In particular $u\in L^{2}$, and
  therefore $u=0$.

\end{proof}
\begin{corollary}
  \label{cor:somm-radi-cond} Suppose $u\in (B^
\mu)^*_0 $ solves the equation $Hu=0$. Then $u=0$.
\end{corollary}
\begin{remark}
  \label{remark:somm-radi-cond} There is a similar characterization of
  the operator $R(0-\i 0)=\lim_{z\to 0,  z\in
    \Gamma(\theta)}R(\bar z)$
  as the one for $R(0+\i 0)$ in Theorem \ref{thm:somm-radi-conds}. We
  are not stating the result  here.
\end{remark}


\end{document}